\documentclass[conference,10pts]{IEEEtran}
\IEEEoverridecommandlockouts

\usepackage{graphicx}

\usepackage[tbtags]{amsmath} 
\usepackage{amssymb}  
\usepackage{verbatim} 
\usepackage{amsxtra}  

\usepackage{enumerate}
\usepackage{amsfonts}
\usepackage{color}
\usepackage{cite}

\usepackage{epsfig, graphicx, psfrag}

\newtheorem{thm}{Theorem}
\newtheorem{lemma}{Lemma}
\newtheorem{remark}{Remark}

\newtheorem{corol}{Corollary}
\newtheorem{prop}{Proposition}
\newtheorem{defn}{Definition}

\providecommand{\thmref}[1]{Theorem~\ref{#1}}

\providecommand{\defnref}[1]{Definition~\ref{#1}}
\providecommand{\secref}[1]{Section~\ref{#1}}

\providecommand{\propref}[1]{Proposition~\ref{#1}}
\providecommand{\remref}[1]{Remark~\ref{#1}}
\providecommand{\figref}[1]{Figure~\ref{#1}}

\providecommand{\appref}[1]{Appendix~\ref{#1}}

\providecommand{\algoref}[1]{Algorithm~\ref{#1}}

\newcommand{\bm}[1]{\mbox{\boldmath{$#1$}}}

\newcommand{\SNR}{\text{SNR}}

\newcommand{\SDR}{\text{SDR}}

\newcommand{\mS}{\mathcal{S}}

\newcommand{\Comment}[1]{}
\newcommand{\old}[1]{}
\newcommand{\rem}[1]{}

\newcommand{\tx}{\tilde x}

\newcommand{\tz}{\tilde z}
\newcommand{\tT}{\tilde T}

\newcommand{\tU}{\tilde U}

\newcommand{\ty}{\tilde y}

\newcommand{\bt}{\bm t}

\newcommand{\by}{\bm y}

\newcommand{\bx}{{\bm x}}
\newcommand{\br}{{\bm r}}
\newcommand{\brho}{{\bm \rho}}

\newcommand{\bz}{{\bm z}}

\newcommand{\bmu}{{\bm \mu}}

\providecommand{\comment}[1]{}

\newcommand{\beqn}[1]{\begin{eqnarray}\label{#1}}
\newcommand{\eeqn}{\end{eqnarray}}
\newcommand{\beq}[1]{\begin{equation}\label{#1}}
\newcommand{\eeq}{\end{equation}}

\newcommand{\trace}[1]{\mathrm{trace}\left( #1 \right)}

\providecommand{\var}[1]{{\rm Var\left( #1 \right)}}

\providecommand{\Ddef}{\triangleq}

\newcommand{\tby}{\tilde {\bm y}}
\newcommand{\tbx}{\tilde {\bm x}}
\newcommand{\tbz}{\tilde {\bm z}}

\newcommand{\digital}{\text{digital}}
\newcommand{\analog}{\text{analog}}

\usepackage{float}

\floatstyle{ruled}
\newfloat{algorithm}{thp}{lop}
\floatname{algorithm}{Algorithm}

\begin{document}

\title{Joint Unitary Triangularization for MIMO Networks}

\author{
\authorblockN{Anatoly Khina}
\authorblockA{Dept. of EE-Systems,\\
Tel-Aviv University \\
Email: anatolyk@eng.tau.ac.il}
\and
\authorblockN{Yuval Kochman}
\authorblockA{EECS Dept., \\
MIT \\
Email: yuvalko@mit.edu}
\and
\authorblockN{Uri Erez}
\authorblockA{Dept. of EE-Systems,\\
Tel-Aviv University \\
Email: uri@eng.tau.ac.il}
\thanks{
The work of the second author was supported in part by the National Science Foundation under grant No. CCF-1017772
and by a grant from Hewlett-Packard Laboratories.

The work of the third author was supported in part by the U.S.~-- Israel Binational Science
Foundation under grant 2008/455.

The material in this paper was presented in part at the 48th Annual Allerton Conference on Communication, Control and Computing, 2010. Another part of this work is to be presented ICASSP~2011, Prague, Czech Republic.}
}

\maketitle




\begin{abstract}
This work considers communication networks where individual links can be described as MIMO channels. Unlike orthogonal modulation methods (such as the singular-value decomposition), we allow interference between sub-channels, which can be removed by the receivers via successive cancellation. The degrees of freedom earned by this relaxation are used for obtaining a basis which is simultaneously good for more than one link. Specifically, we derive necessary and sufficient conditions for shaping the ratio vector of sub-channel gains of two broadcast-channel receivers. We then apply this to two scenarios: First, in digital multicasting we present a practical capacity-achieving scheme which only uses scalar codes and linear processing. Then, we consider the joint source-channel problem of transmitting a Gaussian source over a two-user MIMO channel, where we show the existence of non-trivial cases, where the optimal distortion pair (which for high signal-to-noise ratios equals the optimal point-to-point distortions of the individual users) may be achieved by employing a hybrid digital-analog scheme over the induced equivalent channel. These scenarios demonstrate the advantage of choosing a modulation basis based upon multiple links in the network, thus we coin the approach ``network modulation''.
\end{abstract}

\vspace{1mm}

{\keywords Broadcast channel, MIMO, multicasting, generalized triangular decomposition, geometric mean decomposition, GSVD, GDFE, multiplicative majorization, joint source-channel coding.}

\section{Introduction}
\label{sec:intro}

The choice of modulation domain plays a major role in communication, both in deriving performance limits and in the design of practical
schemes which decouple the signal processing task of channel equalization from coding.
Thus, choosing the ``right" basis is of central importance.
For example, the capacity of the Gaussian inter-symbol interference (ISI) channel is given by the water-filling solution, applied in the frequency domain; the same transformation also allows  to use popular schemes such as Orthogonal Frequency-Division Multiplexing (OFDM) which employs the discrete Fourier transform. The singular value decomposition (SVD) plays a similar role for multiple-input multiple-output (MIMO) channels. Common to both cases is \emph{diagonalization}: They yield parallel independent equivalent channels.
But do we really need such orthogonality?
Capacity can be achieved for both the ISI and MIMO channels using \emph{non-orthogonal} equivalent channels, by a receiver which performs \emph{triangularization} of the channel\footnote{Outside the high signal-to-noise ratio regime, ``near triangularization'' is performed as an optimal balance between residual interference and noise.} (rather than diagonalization) and then decision-feedback equalization or successive interference cancellation (SIC), see e.g.~\cite{Wolniansky_V-BLAST}.  This is done without performing any transformation at the transmitter. It  is therefore natural to ask, what can be achieved by allowing \emph{both} a transmitter transformation (in addition to the receiver one) and SIC.

One such direction, pursued by Jiang, Hager and Li \cite{GTD}, is the \emph{generalized triangular decomposition} (GTD): A matrix $A$ is decomposed as
\begin{align*}
    A = U T V^\dagger \,,
\end{align*}
where $U$ and $V$ are unitary matrices, $V^\dagger$ denotes the conjugate transpose of $V$ and $T$ is upper-triangular with a prescribed diagonal. It is shown in \cite{WeylCondition,WeylConditionInverse_ByHorn} that the transforming matrices $U$ and $V$ exist if and only if the (desired) diagonal elements of $T$ obey Weyl's multiplicative majorization relation with the singular values of $A$ (see also \cite{PalomarJiang}). Since the product of these diagonal elements equals the product of the singular values of $A$, the decomposition performs \emph{diagonal shaping}: It distributes the total gain between the diagonal elements in a desired way. An important special case is where balanced gains are sought, i.e., the diagonal elements of $T$ should all be equal. In that case, named the \emph{geometric mean decomposition} (GMD) \cite{GMD}, the majorization condition holds for any $A$. When applied to MIMO communication, GMD has an advantage over SVD, that all subchannels enjoy the same gain, and thus may support codebooks of the same rate, avoiding the need for a bit-loading mechanism. This comes at the price of performing SIC at the receiver. The GMD has received considerable attention; see, e.g., \cite{GMD-ISI,ShenoudaDavidson,GMD-Transform}  for some of its applications.

We take a different path, in which we wish to jointly shape the diagonals of two matrices, for the purpose of multi-terminal communication. Since with this approach the choice of basis depends upon more than one communication link, we call it \emph{network modulation}. We jointly triangularize two matrices $A_1$ and $A_2$ as
\beq{STUD}
    A_i = U_i T_i V^\dagger \,, \quad i=1,2 \,,
\eeq
where $U_1$, $U_2$ and $V$ are unitary and $T_1$ and $T_2$ are upper-triangular. Having the same matrix $V$ on one of the sides of the decomposition corresponds to applying the same transformation, and is thus suitable to two links originating (or terminating) at the same node.
It turns out that, in different network applications, it is important to shape the vector of \emph{ratios} between the diagonals. We show that the sufficient and necessary condition for achievability of a ratio vector is a multiplicative majorization relation with the generalized singular values \cite{VanLoan76} of the pair $(A_1,A_2)$.

We start by deriving the necessary and sufficient conditions for joint unitary triangularization of two matrices, in \secref{sec:GGTD}. In the rest of the paper we apply this result in two different scenarios, where in one we present an optimal practical scheme for a problem for which the capacity is known, and in the second we derive the (hitherto unknown) optimal performance.
In \secref{sec:MulticastingScheme} we combine the joint triangularization with the concept of SIC, to
present an optimal scheme for two-user digital multicasting that employs linear processing of scalar codebooks.
In \secref{sec:JSCC}, we address the problem of transmission of an \emph{analog} source over two MIMO links,
where we show that a ratios vector of all-ones except for one element creates an equivalent channel over which a hybrid digital-analog (HDA) scheme can achieve the optimal tradeoff between user distortions; thus we derive the optimal performance whenever the channels are such that this ratios vector is feasible. We conclude the paper in \secref{sec:conclusions}.

We note that the decomposition may equally be applied to cases where two transmitters communicate with a joint receiver via MIMO links (a MIMO MAC channel). In this case the roles of the $U$ and $V$ matrices in \eqref{STUD} are interchanged. An application of the decomposition in such a setting is a MIMO extension of the ``physical network coding'' approach to bi-directional relays \cite{WilsonRelays}. This application is beyond the scope of this paper, and appears in \cite{JET:TwoWayRelayISIT11}.


\section{Joint Unitary Triangularization}
\label{sec:UnitaryTriang.}

In this section we present the joint unitary triangular decomposition of two matrices.
We start by introducing a few notations and recalling known decompositions for a single matrix.

\subsection{Unitary Triangularization of a Single Matrix} \label{sec:GTD}
Throughout the work, we will only need to decompose matrices which belong to the following class.

\vspace{3mm}
\begin{defn}[Proper dimensions]
    An $m \times n$ matrix $A$ is said to have proper dimensions if it is full-rank and $m \geq n$.
\vspace{2mm}
\end{defn}

The singular-value decomposition (SVD, see \cite{GolubVanLoan3rdEd}) of a matrix $A$ of proper dimensions $m \times n$, is given by:
\begin{align}
\label{eq:UnitaryTriang.}
    A = U T V^\dagger \,,
\end{align}
where $U$ and $V$ are unitary matrices, and $T$ is an $m \times n$ (generalized) diagonal matrix, viz.\ $T_{i,j} = 0$ for $i \neq j$.

\vspace{2mm}
\begin{remark} \label{rem:phase}
Throughout the paper, we will assume in all decompositions that all the diagonal elements $T_{i,i}$ are real and non-negative. This is without loss of generality, since any phase can be absorbed in $U$ and $V$.
\vspace{2mm}
\end{remark}

The diagonal elements of $T$ are called the the singular values (SV) of $A$; they equal the square-roots of the eigenvalues of $A^\dagger A$. Since we assumed $A$ to be full-rank, all its SVs are strictly positive. We define the SV vector $\bmu(A)$ as the $n$-dimensional vector composed of all SVs (including their algebraic multiplicity), ordered non-increasingly. $\mu$(A) is unique, i.e., there is no other diagonalization, up to ordering and phases of the diagonal.

Unitary triangularization coined generalized triangular decomposition (GTD, see \cite{GTD}) is a generalization of the SVD to triangular matrices $T$. The class of matrices is formally defined as follows.

\vspace{3mm}
\begin{defn}[Square part]
\label{def:square}
    Let $A$ be a matrix of dimensions $m \times n$ where $m \geq n$. The square part of $A$, denoted $[A]$, consists of the first $n$ rows of $A$.
\vspace{2mm}
\end{defn}
\vspace{3mm}
\begin{defn}[Generalized triangular matrix]
    Let $T$ be a matrix of proper dimensions. We call $T$ a generalized triangular matrix, if $T_{i,j} = 0$ for $i > j$, i.e.,
    it has the block structure
    \[
        T = \left( \begin{array}{c}
                    [T] \\
                    0
                   \end{array}
            \right) \]
    where the square part $[T]$ is upper-triangular.
\vspace{2mm}
\end{defn}

As for the SVD, we assume non-negative diagonal elements (see Remark~\ref{rem:phase}), and they are all positive under the full-rank assumption.
Note that for any unitary triangularization,
\beq{triangular_det}
    \det\left(A^\dagger A\right) = \det\left(T^\dagger T\right) = \left(\det [T]\right)^2 = \prod_{j=1}^n (T_{j,j})^2 \,.
\eeq

It turns out, that the singular values are an extremal case for the diagonal of all possible unitary triangularizations. For stating this, we need the following.

\vspace{3mm}
\begin{defn}[Multiplicative majorization (see \cite{PalomarJiang})]
\label{def:major}
    Let $\bx$ and $\by$ be two $n$-dimensional vectors of positive elements. Denote by $\tbx$ and $\tby$ the vectors composed of the entries of $\bx$ and $\by$, respectively, ordered non-increasingly. We say that $\bx$ majorizes $\by$ ($\bx \succeq \by$) if they have equal products:
    \[
        \prod_{j=1}^n  x_j  =  \prod_{j=1}^n y_j  \,,
    \]
    and their (ordered) elements satisfy
for any $1 \leq k < n$,
    \[
        \prod_{j=1}^k  \tx_j  \geq  \prod_{j=1}^k  \ty_j  \,.
    \]
\vspace{2mm}
\end{defn}

In these terms, we can give the condition of the GTD \cite{GTD} for the existence of a unitary triangularization: Let $A$ be a matrix of proper dimensions $m \times n$ and $\bt$ be an $n$-dimensional vector of positive elements. Then, there exists a unitary triangularization of $A$ with diagonal $\bt$, i.e., $A$ can be decomposed as in \eqref{eq:UnitaryTriang.} with $T$ being some generalized upper-triangular matrix with the prescribed diagonal $\bt$, if and only if the latter is majorized by the singular-values vector of $A$:
    \begin{align*}
        \bmu(A) \succeq \bt \,.
    \end{align*}

\subsection{Joint Triangularization with Shaped Diagonal Ratio}
\label{sec:GGTD}

The SVD presented in \secref{sec:GTD} is essentialy unique. Thus, in general, two matrices cannot be jointly diagonalized by unitary matrices. Nevertheless, joint \emph{triagonalization} \eqref{STUD} is possible. In this section we prove the necessary and sufficient conditions for such triangularization, formally defined as follows.

\vspace{3mm}
\begin{defn}[Joint Unitary Triangularization] \label{def:GGTD}
Let $A_1$ and $A_2$ be matrices of proper dimensions with the same number of columns. A joint decomposition:
\begin{align}
\label{eq:JointTriang.}
    \begin{aligned}
        A_1 &= U_1 T_1 V^\dagger \\
        A_2 &= U_2 T_2 V^\dagger \,,
    \end{aligned}
\end{align}
is called a joint unitary triangularization if $U_1$, $U_2$ and $V$ are unitary, and $T_1$ and $T_2$ are generalized upper-triangular matrices of the same dimensions as $A_1$ and $A_2$, respectively.
\vspace{2mm}
\end{defn}

The existence condition for this joint decomposition turns out to be similar to that of the GTD (see \secref{sec:GTD}), where the SVs are replaced by generalized singular values, and the diagonal of $T$ is replaced by the ratio of the diagonals of $(T_1,T_2)$. These quantities are defined below.

\vspace{3mm}
\begin{defn}[Generalized singular values \cite{VanLoan76}]
\label{def:GSV}
    For any (ordered) matrix pair $(A_1,A_2)$, the generalized singular values (GSVs) are the positive solutions $a$ of the equation
    \[
        \det\left\{A_1^\dagger A_1 - a^2 A_2^\dagger A_2 \right\} = 0 \,.
    \]
    Let the GSV vector $\bmu(A_1,A_2)$ be the vector composed of all GSVs (including their algebraic multiplicity), ordered non-increasingly.
\vspace{2mm}
\end{defn}

\begin{remark}
 For matrices of proper dimensions we have $n$ GSVs, all positive (recall Remark~\ref{rem:phase}) and finite.
For non full-rank matrices we still have $n$ GSVs, even if the number of finite solutions is smaller. We define a GSV to be infinite, if the corresponding GSV of the matrices in reversed order is zero. If the number of finite and infinite solutions is smaller than $n$, this suggests that the column rank can be reduced without changing the problem; we shall assume the problem is in its reduced form.
\end{remark}
\vspace{2mm}
\begin{remark}
\label{rem:GSV_coincides_with_SV}
        When $A_1$ and $A_2$ are square and non-singular,  $\bmu(A_1,A_2)$ consists of the singular values of $A_1 A_2^{-1}$ \cite{VanLoan76,GolubVanLoan3rdEd}.
\end{remark}

\vspace{3mm}
\begin{defn}[Diagonal ratios vector]
    Let $T_1$ and $T_2$  be two generalized upper-triangular matrices of proper dimensions $m_1 \times n$ and $m_2 \times n$, respectively, with positive diagonal elements. The diagonal ratios vector $\br(T_1,T_2)=\br([T_1],[T_2])$ is the  $n$-length vector whose $j$-th entry is equal to the diagonal ratio
    $T_{1;j,j} /  T_{2;j,j}$, where $T_{i;j,k}$ denotes the $(j,k)$ entry of $T_i$ ($i=1,2$).
\vspace{2mm}
\end{defn}

We are now ready to prove the main result of this section, giving the condition for existence of joint unitary triangularization in terms of majorization (recall Definition~\ref{def:major}). Since the existence proof is constructive, it results with a decomposition procedure; this is summarized in \algoref{alg:SquareSTUD} for the simpler case of square matrices, and in \algoref{alg:STUD}~-- for the general case.

\begin{algorithm}[t]
    \caption{\textbf{: Joint Triangularization of Square Matrices}}
    \centering
    \fbox{
        \begin{minipage}{.94\linewidth}
            \begin{itemize}\addtolength{\itemsep}{0.5\baselineskip}
            \item
                Compute $B \triangleq A_1 A_2^{-1}$.
            \item
                Apply the GTD to $B$ (for details regarding the implementation, including Matlab code, see \cite{GTD}):
                \begin{align*}
                    B = U_1 R U_2^\dagger \,,
                \end{align*}
                with $R$ having a diagonal equal to (the desired diagonal ratio) $\br(T_1,T_2)$.
            \item
                Decompose $U_i^\dagger A_i$ according to the RQ decomposition:
                \begin{align*}
                    U_1^\dagger A_1 &= T_1 V^\dagger \\
                    U_2^\dagger A_2 &= T_2 V^\dagger \,.
                \end{align*}
            \end{itemize}
        \end{minipage}
    }
\label{alg:SquareSTUD}
\end{algorithm}

\begin{algorithm}[t]
    \caption{\textbf{: Joint Triangularization}}
    \centering
    \fbox{
        \begin{minipage}{.94\linewidth}
            \begin{itemize}\addtolength{\itemsep}{0.5\baselineskip}
            \item
                Apply (individual) QR decompositions to $A_1$ and $A_2$:
                \begin{align*}
                    A_i = Q_i R_i \qquad i=1,2 \,.
                \end{align*}
            \item
                Extract $[R_i]$, the upper square $n \times n$ part of $R_i$.
            \item
                Apply \algoref{alg:SquareSTUD} to $[R_1]$ and $[R_2]$ to obtain a joint decomposition
                \begin{align*}
                    [R_1] &= \tU_1 \tT_1 V^\dagger \\
                    [R_2] &= \tU_2 \tT_2 V^\dagger \,,
                \end{align*}
                where the diagonal ratio vector of $(\tT_1,\tT_2)$ is equal to the desired diagonal ratio, i.e., $\br(\tT_1,\tT_2) = \br(T_1,T_2)$.
            \item
                Construct the matrices:
                \begin{align*}
                T_i =
                    \left(
                      \begin{array}{c}
                	\tT_i \\
                	0 \\
                      \end{array}
                    \right) \,,
                \end{align*}
                where the lower zero block is of dimensions \mbox{$(m_i-n) \times n$}, and
                \begin{align*}
                    U_i = Q_i \left(
                	    \begin{array}{cc}
                	      \tU_i & 0 \\
                	      0 & I_{m_i-n} \\
                	    \end{array}
                	  \right) \,.
                \end{align*}
            \end{itemize}
        \end{minipage}
    }
\label{alg:STUD}
\end{algorithm}

\vspace{3mm}
\begin{thm}
\label{thm:GGTD}
    Let $A_1$ and $A_2$ be matrices of proper dimensions  $m_1 \times n$ and $m_2 \times n$, respectively, and $\br$ be an $n$-dimensional vector with positive elements. Then, there exists a joint unitary triangularization of $(A_1,A_2)$ with diagonal ratio vector $\br$, i.e., the matrices can be decomposed as:
    \begin{align}
    \label{eq:JointTriang.}
        \begin{aligned}
            A_1 &= U_1 T_1 V^\dagger \\
            A_2 &= U_2 T_2 V^\dagger\,,
        \end{aligned}
    \end{align}
       where $U_1$, $U_2$ and $V$ are unitary, and $T_1$ and $T_2$ are some generalized upper-triangular matrices with the prescribed diagonal ratio (the $i$-th element of $\br(T_1,T_2)$ equals $r_i$),
    if and only if the latter is majorized by the GSV vector of $(A_1,A_2)$:
    \begin{align}
    \label{eq:GWeyl}
        \bmu(A_1,A_2) \succeq \br \, .
    \end{align}
   \vspace{2mm}
\end{thm}

\begin{proof} \textbf{Achievability part.}
    We present here a proof for the case when the matrices are square ($m_1\! =\! m_2=n$). The extension to the general proper-dimension case is relegated to \appref{app:Proof_GGTD_non-square}.

    In the square case, $A_1$ and $A_2$ must be invertible, being full-rank.
    Define the matrix $B = A_1 A_2^{-1}$.
    The SV vector of $B$, $\bmu(B)$, coincides with the GSV vector of $(A_1,A_2)$, $\bmu(A_1,A_2)$ (recall \remref{rem:GSV_coincides_with_SV}).
    Thus, it majorizes $\br$, by assumption.
    Hence, according to the GTD (see \secref{sec:GTD}), the matrix $B$ can be decomposed as
    \begin{align}
    \label{eq:A1*inv(A2)=U1R1U2'}
        B = U_1 R U_2^\dagger \,,
    \end{align}
    where $U_1$ and $U_2$ are unitary and $R$ is upper-triangular with a diagonal vector which equals $\br$.
     Now, apply RQ decompositions to $U_i^\dagger A_i$ ($i=1,2$) to achieve
    \begin{align}
    \label{eq:Ai=UiRiVi'}
        U_i^\dagger A_i = T_i V_i^\dagger \,,
    \end{align}
    where $T_i$ are upper-triangular with positive diagonal entries and $V_i$ are unitary.
    By substituting \eqref{eq:Ai=UiRiVi'} into \eqref{eq:A1*inv(A2)=U1R1U2'} we have
    \begin{align*}
        U_1 T_1 V_1^\dagger V_2 T_2^{-1} U_2^\dagger &= U_1 R U_2^\dagger \,,
    \end{align*}
    which is equivalent to
    \begin{align}
    \label{eq:Unitary=Triang}
        V_1^\dagger V_2 = T_1^{-1} R T_2 \,.
    \end{align}
    We note that the l.h.s. of \eqref{eq:Unitary=Triang} is unitary, whereas its r.h.s. is an upper-triangular matrix
    with positive diagonal entries. An equality between such matrices can hold only if both matrices are equal to the identity matrix of the appropriate dimensions ($n \times n$).
    Thus, we have
    \begin{align*}
        V & \triangleq V_1 = V_2 \,, \\
        T_{1;i,i} &= R_{i,i} T_{2;i,i} \,, \qquad i=1,...,n \,.
    \end{align*}

    Since the diagonal of $R$ is equal to $\br$, this establishes the desired decomposition \eqref{eq:JointTriang.}.

    \textbf{Converse part.}
    Assume, in contradiction, that $A_1$ and $A_2$ can be decomposed as in \eqref{eq:JointTriang.} such that $\bmu(A_1,A_2) \nsucceq \br(T_1,T_2)$. Note that $\bmu(T_1,T_2) = \bmu(A_1,A_2)$.
    Moreover, $[T_1]$ and $[T_2]$ are non-singular square matrices of dimensions $n \times n$ with GSV and diagonal-ratios vectors which are equal to those of $(T_1,T_2)$, i.e.,
    \begin{align*}
        \bmu([T_1],[T_2]) & = \bmu(T_1,T_2) = \bmu(A_1,A_2) \,, \\
    \br([T_1],[T_2]) & =  \br(T_1,T_2) \, .
    \end{align*}
    Thus $\bmu([T_1],[T_2]) \nsucceq \br([T_1],[T_2])$, which in turn implies that the upper-triangular matrix \mbox{$B \triangleq [T_1] [T_2]^{-1}$}
    has a diagonal $\br([T_1],[T_2])$ and an SV vector $\bmu([T_1],[T_2])$. But according to Weyl's condition~ \cite{WeylCondition}:
    \begin{align*}
        \bmu(A_1,A_2) = \bmu([T_1],[T_2]) \succeq \br([T_1],[T_2]) = \br(T_1,T_2) \,,
    \end{align*}
    in contradiction to the assumption.
    \vspace{2mm}
\end{proof}
\vspace{2mm}
\begin{remark}
    Note that we did not require the matrices $T_1$ and $T_2$ to satisfy Weyl's condition individually, as we did not strive to design specific diagonal values but rather prescribed \emph{ratios} between the diagonals of $T_1$ and $T_2$. Indeed, one may verify that the resulting matrices in \thmref{thm:GGTD} satisfy Weyl's condition individually.
\end{remark}
\vspace{2mm}
\begin{remark}
    By the unitarity of $U_1$, $U_2$ and $V$, the products of $\bmu$ and $\br$ are equal. Thus, the majorization relations mean that the diagonal ratios are always ``less spread'' than the generalized singular values. This is also true for the (individual) diagonal values of $T_i$ ($i=1,2$) being ``less spread'' than the singular values of $A_i$.
\end{remark}
\vspace{2mm}
\begin{remark}[Relation to GSVD]
\label{rem:GSVD}
    The GSVD \cite{VanLoan76} can be stated in a triangular form \eqref{eq:JointTriang.}, with diagonals ratio $\br(T_1,T_2)=\bmu(A_1,A_2)$. Thus, the GSVD is a limiting case with maximal ratios spread.
\end{remark}
\vspace{2mm}
\begin{remark}[Relation to GTD]
    Taking in the joint decomposition $A_2=I$ yields the GTD of $A_1$ \cite{GTD}; further, the GSV vector becomes the SV vector of $A_1$. The existence condition, in turn, reduces to the Weyl condition (see e.g. \cite{GTD}). In this sense, the condition in Theorem~\ref{thm:GGTD} may be seen as a generalized Weyl condition for joint triangularization.
\end{remark}
\vspace{2mm}
\begin{remark}[Relation to the generalized Schur decomposition]
  This decomposition, also called the QZ-decomposition \cite{GolubVanLoan3rdEd}, is a special case of the joint triangularization \eqref{eq:JointTriang.} with $U_1=U_2$. It can be shown that the diagonal ratio vector induced by this decomposition is unique, i.e., requiring that the unitary matrices are the same on both sides prohibits shaping of the diagonal ratios.
\end{remark}
\vspace{2mm}
\begin{remark}[Fixed diagonal ratio]
\label{rem:FixedRatio}
    Note that any vector $\bmu$ majorizes the vector where all elements equal to its geometric mean. As a consequence, for any two matrices there exists a joint decomposition with fixed diagonal ratio. We use this fact in Section~\ref{subs:Mutlicasting}.
\end{remark}
\vspace{2mm}
The joint unitary triangularization (and, as a special case, the GTD) can also be relaxed to a block form:
\begin{align}
\label{eq:BlockTriangular}
    T_i = \left(
          \begin{array}{cccc}
            T_{i;1,1} & T_{i;1,2} & \cdots & T_{i;1,K} \\
            0 & T_{i;2,2} & \cdots & T_{i;2,K} \\
            \vdots & \ & \ddots & \vdots \\
            0 & \cdots & 0 & T_{i;K,K} \\
            0 & \cdots & 0 & 0 \\
          \end{array}
        \right) \,,
\end{align}
where $T_{i;k,l}$ is a block of dimensions $n_k \times n_l$, such that $\sum_{k=1}^K n_k = n$  (thus the last row of blocks consists of $n-m$ all-zero rows).
Note that we require corresponding blocks in both matrices to be of the same dimensions, except for the last row of zero blocks.

The existence condition can be stated as the following extension of Theorem~\ref{thm:GGTD}.
Denote by $\left\{ k_l \right\}_{l=1}^K$ the indices satisfying:
\begin{align*}
    r_{k_1} \geq r_{k_2} \geq \cdots \geq r_{k_K} \,,
\end{align*}
where $r_k \triangleq \sqrt[n_k]{\left| \det\left( T_{1;k,k} \right) / \det\left( T_{2;k,k} \right) \right|}$.
Denote by $\brho$ the $K$-length vector composed of $\left\{r_k^{n_k} \right\}$ (i.e., $\left\{\rho_k\right\}_{k=1}^K$ the absolute values of the ratios between determinants of corresponding blocks), ordered in non-increasing order of $\{r_k\}$.
Further denote by $\bmu$ the $K$-length vector composed of the products of sizes $\{k_l\}_{l=1}^K$ of the GSVs of $(A_1,A_2)$ ordered non-increasingly, i.e., the first entry of $\bmu$ is the product of the largest $k_1$ GSVs, its second entry is the product of the next $k_2$ GSVs, etc. Then, the desired block triangularization is possible if and only if the products of the entries of $\brho$ and of the GSV vector are equal,
and for any $1 \leq k < K$:\footnote{These conditons are similar to the majorization conditions of \defnref{def:major} up to the ordering which is done w.r.t. $r_k$ and not those of the entries of $\brho$.}
\begin{align*}
    \prod_{l=1}^k \rho_l \leq \prod_{l=1}^k \mu_l \,.
\end{align*}

\section{Transmission Scheme for Multicasting}
\label{sec:MulticastingScheme}

In this section we derive an optimal \emph{practical} communication scheme for two-user multicasting over Gaussian MIMO BC channels. We start by recalling how the single-user Gaussian MIMO capacity may be achieved using multiple codebooks (each designed for a scalar AWGN channel) with SIC over an equivalent channel, obtained by unitary triangularization (as in Section~\ref{sec:GTD}).

\subsection{SIC for MIMO Channels}
\label{subs:SingleUserSIC}

The exposition below follows that of the universal matrix decomposition (UCD) \cite{UCD}, which is in turn based upon the derivation of the MMSE version of V-BLAST, see, e.g., \cite{HassibiVBLAST}.
Later in the paper we take the triangularization to be one which is simultaneously good for two users. This is suppressed for now. We assume throughout the paper perfect channel knowledge everywhere.

We consider a point-to-point (complex) MIMO channel:
\beq{MIMO}
    \by = H \bx + \bz,
\eeq
where $\bx$ is the channel input of dimensions $N_t \times 1$ subject to an average power constraint $P$;\footnote{Alternatively, one can consider an input
covariance constraint $C_\bx \triangleq E\left[ \bx \bx^\dagger
\right] \preceq C$, where by $C_1 \preceq C_2$ we mean that the
matrix $\left(C_2 - C_1\right)$ is positive semi-definite.} $\by$ is the channel output vector of dimensions
$N_r \times 1$; $H$ is the channel matrix of dimensions $N_r \times N_t$ and $\bz$ is an additive circularly-symmetric Gaussian noise vector of dimensions $N_r \times 1$. Without loss of generality, we assume that the noise elements are mutually-independent, identically-distributed with unit variance.

The capacity of this channel is given by
\beq{MIMO_capacity}
    C(H,P) = \max_{C_\bx} I(H,C_\bx) ,
\eeq
where the maximization is over all channel input covariance matrices $C_\bx \geq 0$, subject to the power constraint $\trace{C_\bx}\leq P$, and \footnote{All logarithms in this paper are to the base 2.}
\beq{MIMO_MI}
  I(H,C_\bx) \triangleq \log \det \left( I + H C_\bx H^\dagger \right) \,.
\eeq
We may interpret $I(H,C_\bx)$ as the maximal mutual information that can be attained using an input covariance matrix $C_\bx$, which is achievable by a Gaussian input $\bx$.

In order to achieve a rate approaching this mutual information, optimal codes of long block length are needed. However, as pointed out in the introduction, we take an approach which decouples the signal-processing aspects from these of coding. We thus omit the time index throughout the paper; for example, when referring to an input vector $\bx$, we mean the input at any time instant within the coding block. In a practical setting using encoder/decoder pairs of some given quality, one may easily bound the error probability of the scheme using the parameters of the codes.

\begin{figure}[t]
        \psfrag{&E1}{Codebook 1}
        \psfrag{&E2}{\hspace{-.25cm} Codebook $N_t$}
        \psfrag{&V}{$C_\bx^{1/2} V$}
        \psfrag{&b1}{$b_1$}
        \psfrag{&b2}{$b_{N_t}$}
        \psfrag{&tX1}{$\tilde x_1$}
        \psfrag{&tX2}{$\tilde x_{N_t}$}
        \psfrag{&X1}{$x_1$}
        \psfrag{&X2}{$x_{N_t}$}
        \psfrag{&info}{Info.}
        \psfrag{&bits}{bits}
        \psfrag{&split}{Splitter}
        \centering
        \epsfig{file = 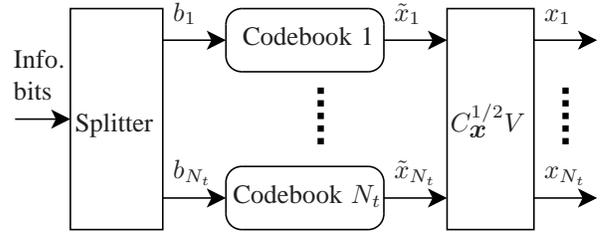, scale = .42}
        \caption{Multiple-codebook transmitter with linear precoding}
        \label{fig:TX}
\end{figure}

When coding over a domain different than the input domain (e.g., time or space), one may start with a virtual input vector $\tbx$, related to the physical input by the linear transformation: \footnote{The square root of a Hermitian positive-definite matrix $A$, denoted by $A^{1/2}$, is defined as the matrix satisfying: $A = A^{1/2} \left( A^{1/2} \right)^\dagger$.}
\beq{MIMO_TX}
    \bx = C_\bx^{1/2} V  \tbx \,.
\eeq
We form the vector $\tbx$, in turn, by taking one symbol from each of $N_t$ parallel codebooks, of equal power $1/N_t$. The matrix $V$ is a unitary linear precoder. See \figref{fig:TX}.

Recalling the GTD (see \secref{sec:GTD}), one may suggest to choose $V$ by applying a unitary triangularization to
\beq{F}
    F \Ddef H C_\bx^{1/2} \,.
\eeq
After the receiver applies the transformation $U^\dagger$, it is left with an equivalent triangular channel $T$, over which it may decode the codebooks using SIC. Unfortunately, while this ``conserves'' the determinant of $H C_\bx H^\dagger$, it fails to do so when the identity matrix is added as in the mutual information $I(H,C_\bx)$ \eqref{MIMO_MI}. Thus, this is optimal in the high SNR limit only, and an MMSE variation is needed in general, as next described.

We start by applying a unitary triangularization (as in \secref{sec:GTD}) to an augmented matrix:
\beq{augmented}
    \left( \begin{array}{c}
                F \\
                I_{N_t}
           \end{array}
    \right)
    \Ddef G  = U T V^\dagger,
\eeq
where $I_{N_t}$ is the identity matrix of dimensions $N_t$. Note that, by construction, $G$ is of proper dimensions, regardless of the dimensions and rank of the channel matrix $H$. That is, it has dimensions $(N_t+N_r) \times N_t$ and is \emph{full-rank}. The square matrices $U$ and $V$ have dimensions $N_t+N_r$ and $N_t$, respectively.
This allows to decompose the total rate into terms associated with the diagonal values of the matrix $T$, as follows:
\begin{align}
\label{eq:by_def}
    I(H,C_\bx) &= \log \det \left( I_{N_r} + F F^\dagger  \right) \\
    \label{eq:I_Sylvester}
    &= \log \det \left( I_{N_t} + F^\dagger F  \right) \\
\label{eq:F_to_G}
    & = \log \det \left( G^\dagger G \right) \\
\label{eq:G_to_sum}
    & = \sum_{j=1}^{N_t} \log (T_{j,j})^2 = \sum_{j=1}^{N_t} R_j \,,
\end{align}
where \eqref{eq:by_def} follows by the definitions \eqref{MIMO_MI}  and \eqref{F}, \eqref{eq:I_Sylvester} is justified by Sylvester's determinant Theorem (see e.g. \cite{TeletarMIMO}), \eqref{eq:F_to_G} is a direct application of the definition \eqref{augmented},
and in \eqref{eq:G_to_sum} we define $R_j \triangleq \log \left( T_{j,j} \right)^2$.
Using the matrices obtained by this decomposition, the following scheme communicates scalar codebooks of rates $\{R_j\}$.

\begin{figure}[t]
       \psfrag{&x1}{$\hat{\tilde x}_{N_t}$}
        \psfrag{&x2}{$\hat{\tilde x}_{N_t-1}$}
        \psfrag{&x3}{$\hat{\tilde x}_2$}
        \psfrag{&x4}{$\hat{\tilde x}_1$}
        \psfrag{&y1}{$y_{N_r}$}
        \psfrag{&y2}{$y_{N_r-1}$}
        \psfrag{&y3}{$y_1$}
        \psfrag{&ty1}{$\ty_{N_t}$}
        \psfrag{&ty2}{$\ty_{N_t-1}$}
        \psfrag{&ty3}{$\ty_1$}
        \psfrag{&y'1}{$y'_{N_t}$}
        \psfrag{&y'2}{$y'_{N_t-1}$}
        \psfrag{&y'3}{$y'_1$}
        \psfrag{&U}{$U^\dagger$}
        \psfrag{&-}{$-$}
        \psfrag{&D1}{Dec.~$N_t$}
        \psfrag{&D2}{Dec.~$N_t-1$}
        \psfrag{&D3}{Dec.~1}
        \psfrag{&G0}{$\tilde T_{N_t-1,N_t}$}
        \psfrag{&G1}{$\tilde T_{1,N_t}$}
        \psfrag{&G2}{$\tilde T_{1,N_t-1}$}
        \psfrag{&G3}{$\tilde T_{1,2}$}
         \centering
        \epsfig{file = 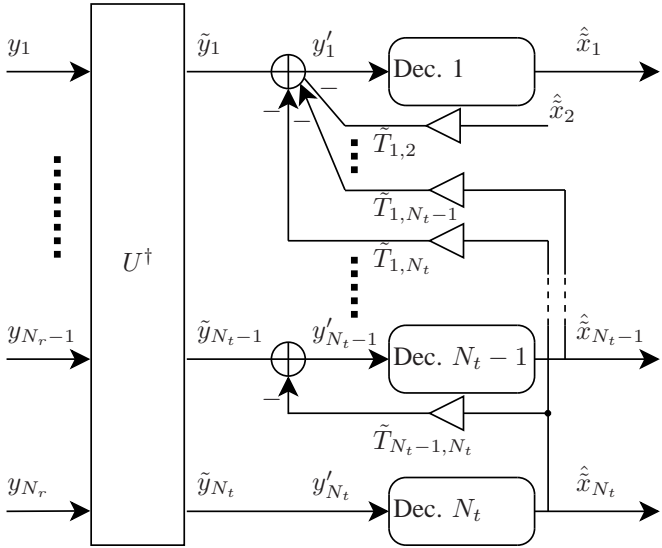, scale = .45}
        \caption{SIC-based receiver.}
        \label{fig:RX}
\end{figure}

The transmitted signal is formed using \eqref{MIMO_TX}.
At the receiver, we use a matrix $W$, consisting of the upper-left $N_r\times N_t$ block of $U$: $\tilde\by = W^\dagger \by$. This results in an equivalent channel:
\beq{equivalent}
    \tilde\by = W^\dagger (F V \tilde \bx + \bz) = W^\dagger F V \tilde \bx + W^\dagger \bz \Ddef \tilde T \tbx + \tilde\bz.
\eeq
Note that since $W$ is not unitary, the statistics of $\tbz \triangleq W^\dagger \bz$ differ from those of $\bz$. We denote the covariance matrix of the equivalent noise by $C_{\tilde\bz} = WW^\dagger$. Finally, SIC is performed, i.e., the codebooks are decoded from last ($j=N_t$) to first ($j=1$), where each codebook is recovered from:
\begin{align}
\label{eq:successive}
    y'_j = \tilde y_j - \sum_{l=j+1}^{N_t} \tilde T_{j,l} \hat {\tilde x}_l \,,
\end{align}
where $\hat {\tilde x}_l$ is the decoded symbol from the $l$-th codebook; see \figref{fig:RX}. Assuming correct decoding of ``past'' symbols, i.e. $\hat {\tilde x}_l = \tilde x_l$ for all $l>j$, the scalar channel for decoding of the $j$-th codebook
is given by:
\beq{scalar_channel}
    y'_j = \tT_{j,j} \tilde x_j + \sum_{l=1}^{j-1}  \tilde T_{j,l} \tilde x_l + \tilde z_j \,.
\eeq
Since $\tilde T$ is not triangular, the second term in this scalar channel (resulting from elements below the diagonal of $\tT$) acts as interference. The signal-to-intereference-and-noise ratio (SINR) is given by:
\begin{align}
\label{eq:SINR}
    S_j = \var{\tilde x_j \middle| \tilde {\bf{y}}, \tilde x_{j+1}^{N_t}}
    = \frac{(\tilde T_{j,j})^2}{C_{\tilde\bz;j,j} + \sum_{l=1}^{j-1} ( \tilde T_{j,l})^2} \,,
\end{align}
where $C_{\tilde\bz;i,j}$ denotes the $(i,j)$ entry of $C_{\tilde\bz}$.

The following, which is equivalent to Lemma III.3 in \cite{UCD}, shows optimality of the scheme.

\vspace{3mm}
\begin{prop}
\label{prop_MMSE}
    For any channel $H$ and input covariance matrix $C_\bx$, the SINRs $S_j$ \eqref{eq:SINR} of the transmission scheme above satisfy:
    \beq{SINRs_rates}
        \log (1 + S_j) = R_j  \,, \qquad \forall j=1,\ldots,N_t \,,
    \eeq
    where the rates $R_j$ are given by \eqref{eq:G_to_sum}.
\vspace{2mm}
\end{prop}

This completes the recipe for a digital transmission scheme which achieves $I(H,C_\bx)$: For a given input covariance matrix $C_\bx$, choose the individual codebook rates to approach $\{R_j\}$, the sum of which equals the mutual information afforded by the MIMO channel \eqref{MIMO_MI}.
By \propref{prop_MMSE}, the successive decoding procedure will succeed with arbitrarily low probability of error for these rates (asymptotically for high-dimensional optimal scalar AWGN codes). Taking $C_\bx$ be the covariance matrix maximizing \eqref{MIMO_capacity}, capacity can be achieved.

The above exposition proves the optimality of the ``scalar coding''  approach~-- the combination of scalar AWGN codebooks, linear processing, and SIC. This approach offers reduced complexity and easy-to-analyze performance when the channel is known at both ends (``closed loop'').
Indeed, special cases of this approach have been suggested and used. In particular, using the SVD results in a \emph{diagonal} equivalent channel matrix $T$, establishing parallel virtual AWGN channels (no SIC is needed), see \cite{TeletarMIMO}. Other schemes, such as generalized decision feedback equalization (GDFE) and Vertical Bell-Laboratories Space-Time coding (V-BLAST), see \cite{CioffiForneyGDFE,Wolniansky_V-BLAST}, are based on the QR decomposition. These do not require linear precoding, i.e., $V=I$. The UCD \cite{UCD} uses both a linear precoder and SIC, in order to achieve $T$ with diagonal elements that are all equal.

\begin{remark}[Number of codebooks]
  If desired, one may work with any number of codebooks above $N_t$, as stated in \cite{UCD}. To see that, add ``virtual transmit antennas'' with corresponding zero channel gains. The capacity remains unchanged, and the optimal channel input covariance matrix will not allocate power to these ``antennas''. The number of codebooks is equal to the number of antennas, including the additional virtual ones.
\end{remark}

All of these schemes have significant advantages over direct capacity-achieving implementation for MIMO channels. Such high-complexity schemes,
e.g., using bit-interleaved coded modulation (BICM) in conjunction with sphere detection, essentially require the  same resources as if working in an ``open loop'' mode.
Thus, the complexity involved is similar to that required for approaching the isotropic mutual information of the channel, when only the receiver knows the channel. We note that this advantage comes at the price of suffering from \emph{error propagation} between the codebooks. This effect has been analyzed and simulated in many works, see e.g., \cite{JiangVBLASTAnalysis,VBLAST_LDPC}.

We conclude this section by pointing out a simple extension to a unitary transformation which induces a \emph{block-traingular} matrix rather than a strictly triangular one. That is, if the matrix $R$ in \eqref{augmented} is of the block generalized upper-triangular form \eqref{eq:BlockTriangular},
where the block $T_{i;k,l}$ is of dimensions $N_k \times N_l$, such that $\sum_{k=1}^K N_k = N_t$.
In that case, we employ $K \leq N_t$ codes in parallel, each over an equivalent $N_k \times N_k$ MIMO channel, achieved by ``block-SIC'':
\beq{block_channel}
    \by'_j = \sum_{l=1}^j  \tilde T_{j,l} \tbx_l + \tilde \bz_j \,, \qquad j=1,\ldots,K \,,
\eeq
where $\tilde T_{j,l}$ is of dimensions $N_j \times N_l$. Seen as Gaussian MIMO channels (i.e., seeing residual interference as noise) we achieve, as an extension to Proposition~\ref{prop_MMSE}, a rate
\beq{eq:block_rates}
    R_j = \log \det \left( T_{j,j} (T_{j,j})^\dagger \right)
\eeq
over each such block channel.

\subsection{Optimal Two-User Scheme}
\label{subs:Mutlicasting}

\begin{figure}[t]
        \psfrag{&T1}{Rx 1}
        \psfrag{&T2}{Rx 2}
        \psfrag{&R}{Tx}
        \psfrag{&H1}{$H_1$}
        \psfrag{&H2}{$H_2$}
        \centering
        \epsfig{file = 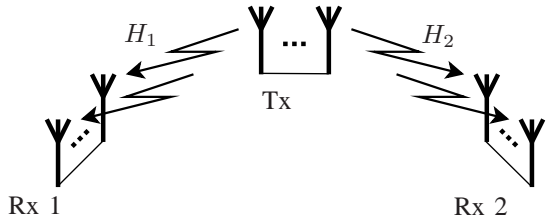, scale = .6}
    \caption{Two-user MIMO BC channel. Even though both users may have similar channel quality, the actual channel matrices $H_1$ and $H_2$ differ.}
    \label{fig:BC_Channel}
\end{figure}

We now derive an optimal practical communication scheme for two-user multicasting.

The two-user Gaussian MIMO broadcast (BC) channel has one transmit and two receive nodes, where each received signal is related to the transmitted signal through (see also \figref{fig:BC_Channel}):
\beq{MIMO-BC}
  \by_i = H_i \bx + \bz_i \,, \qquad i=1,2 \,,
\eeq
where $\bx$ is the channel input of dimensions $N_t \times 1$ subject to an average power constraint $P$;\footnote{Again, alernatively, one can consider an input
covariance constraint $C_\bx \triangleq E\left[ \bx \bx^\dagger
\right] \preceq C$.}
$\by_i$ is the channel output vector of decoder $i$ ($i=1,2$) of dimensions
$N_r^{(i)} \times 1$; $H_i$ is the channel matrix to user $i$ of dimensions $N_r^{(i)} \times N_t$ and $\bz_i$ is an additive circularly-symmetric Gaussian noise vector of dimensions $N_r^{(i)} \times 1$, where again, without loss of generality, we assume that the noise elements are mutually-independent, identically-distributed with unit variance.

This channel has received much attention over the past decade. Unlike the single-input single-output (SISO) case, the Gaussian MIMO BC channel is not degraded. Nevertheless, capacity regions were established for some scenarios, such as private-messages only, and for a common message with a single private message, and bounds were derived for others,
see~\cite{YCdpc,CaireShamai03,VJG03,WSS06,WSS_CommonMessage} and references therein.

We focus our attention on the multicast (common-message) problem, the capacity of which is long known to equal the (worst-case) capacity of the
compound channel~\cite{BlackwellBreimanThomasian59,Dobrushin59,Wolfowitz60}, with the compound parameter being the channel matrix index:
\beq{MIMO_BC_capacity}
    C(H_1,H_2,P) = \max_{C_\bx} \min_{i=1,2} I(H_i,C_\bx) \,,
\eeq
where maximization is over all channel input covariance matrices $C_\bx \geq 0$, subject to the power constraint $\trace{C_\bx}\leq P$ and the MIMO mutual information $I(H,C_\bx)$ was defined in \eqref{MIMO_MI}.

We wish to use a scalar-coding approach, as applied to the point-to-point setting in \secref{subs:SingleUserSIC}. Indeed, the private-message MIMO BC capacity can be achieved by scalar-coding  (in this case dirty-paper coding) techniques; see, e.g., \cite{Tejera05}.
In the presence of a common message, however, to our knowledge, no scalar capacity-approaching coding solutions are known. QR-based schemes fail, since requiring the individual streams to be simultaneously  decodable  at all the receivers implies that the rate \emph{per stream} is governed by the smallest of the corresponding diagonal elements (in the resulting two matrices), (potentially) inflicting an unbounded rate penalty. Adapting SVD to this scenario has an additional problem: The decomposition requires multiplying by a channel-dependent (unitary) matrix at the transmitter, which prevents from using this decomposition for more than one channel simultaneously.\footnote{Indeed, the GSVD allows to use a single transformation for two different channels at one of the ends, but for each virtual parallel channel it yields a different gain for each user, thus not solving the inefficiency mentioned above. In fact, using GSVD may result in \emph{worse} performance than using a QR-based receiver without any transformation at the transmitter since the spread of the diagonal ratio is maximal, see Remark~\ref{rem:GSVD}.}
As a result of these difficulties, other techniques were proposed, which are suboptimal in general, see, e.g., \cite{Gohary03,LopezPhD}.

In this section, we present an optimal successive-decoding (low-complexity) scheme for a two-user common-message Gaussian MIMO BC channel. Specifically, the proposed scheme is based upon  SIC and good \emph{scalar} AWGN codes, in conjunction with
the following special case of the  decomposition in Theorem~\ref{thm:GGTD}.

\vspace{3mm}
\begin{corol}
\label{cor:JET}
  Let $A_1$ and $A_2$ be two matrices of proper dimensions $m_1 \times n$ and $m_2 \times n$, respectively, satisfying
\beq{corol_cond}
    \det\left(A_1 A_1^\dagger\right) \geq \det\left(A_2 A_2^\dagger\right).
\eeq Then there exists a joint triangularization \eqref{eq:JointTriang.} where
\[
  T_{1;j,j} \geq T_{2;j,j} \,, \qquad \forall j=1, \ldots, n \,.
\]
\vspace{2mm}
\end{corol}

\begin{proof}
  An equivalent condition to \eqref{corol_cond} is that the product of the entries of $\bmu=\bmu(A_1,A_2)$ is at least one. Let $\bar{\mu} \geq 1$ be the geometrical mean of $\bmu$, and let the vector $\br$ be the same size as $\bmu$, with all the elements equal to $\bar{\mu}$. By construction, $\bmu \succeq \br$, thus by Theorem~\ref{thm:GGTD} there exists a joint triangularization with this ratio. Consequently, there exists a decomposition such that for all elements $T_{1;j,j} = \bar \mu T_{2;j,j} \geq T_{2;j,j}$.
\end{proof}
\vspace{2mm}
\begin{remark}[Admissible diagonal ratios]
  The proof suggests that the diagonal ratios vector be made uniform. This is always possible, but is not the only choice (unless $I(H_1,C_\bx)=I(H_2,C_\bx)$).
\end{remark}

The results above defines a communication scheme in the following way. For the channels $H_1$ and $H_2$, let $C_{\bx}$ be a capacity-achieving input covariance matrix, and assume without loss of generality that $I(H_1,C_\bx)\geq I(H_2,C_\bx)$. Define the augmented matrices $G_1$ and $G_2$ as in \eqref{augmented}. By Corollary~\ref{cor:JET}, there exists a joint triangularization \eqref{eq:JointTriang.} such that each diagonal element of $[T_1]$ is at least equal to the corresponding element of $[T_2]$. On account of \eqref{eq:by_def}-\eqref{eq:G_to_sum} we have that:
\[
  I(H_2, C_\bx) = \sum_{j=1}^{N_t} \log (T_{2;j,j})^2 \triangleq  \sum_{j=1}^{N_t} R_j \,.
\]
This rate can be approached using SIC at each receiver as in the point-to-point case of \secref{subs:SingleUserSIC}. Specifically, $\tbx$ is formed from $N_t$ codebooks of rates $\{R_j\}$ and power $1/N_t$ each. The transmitted vector is given by the linear precoding \eqref{MIMO_TX} and receiver $i$ performs the linear transformation \eqref{equivalent} and SIC \eqref{eq:successive} (substituting $U_i$ and $T_i$ for $U$ and $T$, respectively). Now Proposition~\ref{prop_MMSE} guarantees correct decoding of all codebooks for receiver 2. Since in receiver 1 each SINR can only be greater, it will be able to decode as well.

\begin{remark}[Private messages]
  If, in addition to the common message intended to both users, there are private messages (messages intended for individual users), superposition may be used. That is, part of the transmit power is dedicated to the private messages. Then, for the purpose of the common message, the transmission for the private messages is considered as noise. This approach was shown in \cite{WSS_CommonMessage} to be capacity-achieving in the presence of a single private message, and under some conditions on the rate and power~-- also in the presence of two private messages (even when these conditions do not hold, superposition gives the best known performance). The scheme presented in this section may be used for the common-message layer of these superposition schemes as well. Interestingly, in that case we would have interference cancellation both at the encoder (dirty-paper coding of the private messages) and at the decoders (SIC of the common message).
\end{remark}

\section{HDA Transmission for Source Multicasting}
\label{sec:JSCC}

In this section we turn from the purely digital setting to a joint source-channel coding (JSCC) problem, where we wish to multicast an analog source to two destinations, where each destination should enjoy reconstruction quality according to the capacity offered by its channel.

The transmission of a source over a BC channel is one of the main applications of JSCC. In this setting, JSCC may be greatly superior to transmission based upon source-channel separation. In a classical example, a white Gaussian source needs to be transmitted over a two-user AWGN BC channel, with one channel use per source sample, under mean-squared error (MSE) distortion. Analog transmission achieves the optimal performance for each user as if the other user did not exist\cite{Goblick65}. In contrast, the separation-based scheme (concatenation of successive-refinment quantization and broadcast channel coding) yields a tradeoff, where if we wish to be optimal for the user with worse signal-to-noise ratio (SNR), then both users have the same distortion, while optimality for the user with better SNR means that the distortion for the other user is trivial (equals the source variance). See, e.g., \cite[App.~A]{ChenWornell98}.

We focus on the transmission of  an i.i.d. circularly-symmetric Gaussian source $S$ to two destinations over a MIMO BC channel \eqref{MIMO-BC}, with one channel use per source sample. We measure the quality of the reproductions $\hat S_i$ using the MSE distortion measure. Thus, we wish to maximize the tradeoff between the signal-to-distortion ratios (SDRs), defined as
\beq{SDRs}
    \SDR_i \Ddef \frac{\var{S}}{\var{\hat S_i - S}} \,, \quad i=1,2 \,.
\eeq
The achievable SDR region $\mS(H_1,H_2)$ is defined as the closure of all pairs which can be achieved by some encoding-decoding scheme.

This general problem of describing $\mS(H_1,H_2)$ has not received much attention. Nevertheless, in the special cases of diagonal or Toeplitz channel matrices, it reduces to the better known problem of transmission over a colored and/or bandwidth-mismatched Gaussian BC channel, for which different schemes which outperform the separation approach have been presented, see e.g. \cite{ChenWornell98,MittalPhamdo,PrabhakaranPuriRamchandran,AnalogMatching,TaherzadehKhandaniNice}. However, even for these cases optimality claims are not abundant. In \cite{AnalogMatching}, Kochman and Zamir show asymptotic optimality for high SNR, where the channels have the same bandwidth as the source, and one user enjoys a better channel than the other at all frequencies. Taherzadeh and Khandani \cite{TaherzadehKhandaniNice} show that optimality in the slope sense (weaker than high-SNR asymptotic optimality) is possible for white channels where their bandwidths (BW) are integer multiples of the source BW. A similar slope argument applies to the general MIMO case as well. 

A simple outer bound on the achievable SDR region is given by the following.
\vspace{3mm}
\begin{prop}
\label{prop:bound}
    $\mS (H_1,H_2) \subseteq \bar{\mS}(H_1,H_2)$, where the bounding region $\bar{\mS}(H_1,H_2)$ is given by:
    \[
        \bigcup_{C_\bx} \big\{(\SDR_1,\SDR_2): \log(\SDR_i) \leq  I(H_i,C_\bx)\big\} \,,
    \]
    where the union is over all matrices $C_\bx \geq 0$ such that $\trace{C_\bx}\leq P$, and where the MIMO mutual information $I(H,C_\bx)$ was defined in \eqref{MIMO_MI}.
    \vspace{2mm}
\end{prop}
The proof follows that of the classical source-channel converse \cite{ShannonRDF}, taking into account that both users share the same channel input.

In \secref{sub:achievability} we find sufficient conditions for achieving points on the boundary of this region. Then, in \secref{sub:rank2} we present, for the case of two transmit antennas, a simple sufficient condition such that all of the region $\bar{\mS}(H_1,H_2)$ can be achieved. Unlike previous work, this proves strict, optimality, non-asymptotic in the channel SNR; it applies to some cases of color and bandwidth mismatch, although not to the white BW-expansion case.

\subsection{Optimality by HDA Transmission} \label{sub:achievability}

We give a constructive achievability proof, which combines a hybrid digital-analog (HDA) scheme by Mittal and Phamdo \cite{MittalPhamdo} with the joint triangularization approach; the optimum is achievable whenever the diagonal ratio vector can be shaped according to the needs of the HDA scheme. In order to understand the function of the HDA scheme, we need to consider the following related scenario. In a JSCC multicasting problem as above, the BC channel is SISO, i.e., $N_t=N_r=1$, in which case the channel matrices reduce to scalars $h_i$. However, in addition, the transmitter node may send some data to the users (identical for both) over a digital channel of rate $R_\digital$ bits per use of the BC channel. According to the aforementioned HDA approach, the source is first quantized according to the rate of the digital channel and the quantization error is then sent in an analog manner over the analog channel. Thus, the distortion is equal to that of the (analog) quantization error, and hence optimality (optimum distortion over each channel) is achieved simultaneously, as implied by the following proposition.

\vspace{3mm}
\begin{prop}
\label{prop:Mittal}
  In the setting above, the optimal performance is given by:
  \[
    \SDR_i = (1+ h_i^2 P) \, 2^{R_\digital} \,.
  \]
\vspace{2mm}
\end{prop}

\begin{proof}
    We use a vector quantizer which decomposes each sample of the Gaussian source $S$ as
    \beq{quantization}
      S = \tilde S + Q.
    \eeq
    The first term is the quantized source, while the second is the quantization error. By quadratic-Gaussian rate-distortion theory (see e.g. \cite{CoverBook}), in the limit of high quantizer dimension, a quantizer of rate $R_\digital$ may achieve:
    \begin{align}
      R_\digital = \log \left( \frac{\var{S}}{\var{Q}} \right) \,,
    \end{align}
    which is equivalent to
    \[
      \SDR_\digital \triangleq \frac{\var{S}}{\var{Q}} = 2^{R_\digital} \,.
    \]
    Now the quantizer output representing $\tilde S$ is sent over the digital channel, thus $\tilde S$ can be reconstructed exactly. Given $\tilde S$, the reconstruction error of $S$ becomes that of $Q$. That is,
    \begin{align*}
        \SDR_i = \frac{\var{S}}{\var{\hat Q_i-Q}} & =  \frac{\var{Q}}{\var{\hat Q_i-Q}} \cdot \SDR_\digital \\
        &\triangleq \SDR_{\analog,i} \; \SDR_\digital \,,
    \end{align*}
    where $\hat Q_i$ is the reconstruction of $Q$ at receiver $i$ using the SISO BC channel. Finally by \cite{Goblick65}, analog transmission of  $Q$ achieves $\SDR_{\analog,i} =  1+ h_i^2 P$, yielding the desired SDRs. No scheme can achieve better performance, by considerations similar to those leading to Proposition~\ref{prop:bound}.
\end{proof}

We use this HDA approach to prove the following.

 \vspace{3mm}
\begin{thm}
\label{thm:JSCC}
 Denote by $\bmu$ the GSV vector of the augmented matrices \eqref{augmented} of the channels with some input covariance matrix $C_\bx$.
   If
      \beq{JSCC_condition}
            \prod_{j=1}^{N_t} \bmu_j \leq 1 \leq \prod_{j=1}^{N_t-1} \bmu_j \,,
        \eeq
   then any pair $(\SDR_1,\SDR_2)$, such that $\log \SDR_i\leq I(H_i,C_\bx)$, is achievable.
\vspace{2mm}
\end{thm}

\begin{proof}
    By \thmref{thm:GGTD}, the condition \eqref{JSCC_condition} implies that there exists a joint unitary triangularization of the augmented channel matrices \eqref{augmented} with diagonal ratios vector which is all one except for one element. The diagonal of $T_i$ can thus be made to satisfy
    \[
      T_{1;j,j} = T_{2;j,j} \triangleq t_j \,, \qquad \forall j=2,\ldots,N_t \,.
    \]
    If we were to send digital data over the MIMO BC channel using this particular triangularization, then by \eqref{eq:G_to_sum} we could send over these $N_t-1$ channels a rate of:
    \[
      R_\digital \triangleq \sum_{j=2}^{N_t} R_j = \sum_{j=2}^{N_t} \log t_j^2 \,.
    \]
    This does not change if we replace, in the transmission scheme, $\tilde x_1$ by a different signal of the same variance $P/N_t$ (since in the SIC process, the decoding of each codebook only depends on decoding of codebooks with \emph{higher} index). Furthermore, regardless of the signal $\tilde x_1$, if the codebooks of subchannels $2,\ldots,N_t$ are correctly decoded then receiver $i$ can obtain the equivalent channel (recall \eqref{scalar_channel}):
    \[
      y'_{i;1} = \tilde T_{1;1,1} \tx_1 + \tz_{i;1} \,,
    \]
    which has, by Proposition~\ref{prop_MMSE}, a signal-to-noise ratio of
    \[\SNR_{\analog,i} = (T_{i;1,1})^2 - 1 \,. \]
    At this stage we have turned the MIMO BC channel into the combination of a digital channel of rate $R_\digital$ and a SISO BC channel of signal-to-noise ratios $\SNR_{\analog,i}$ ($i=1,2$). On account of Proposition~\ref{prop:Mittal}, one can achieve
    \begin{align*}
      \log \SDR_i &= \log (1+\SNR_\analog) + R_\digital \\
      &= \sum_{j=1}^{N_t} \log (T_{i;j,j})^2 \\
      &= I(H_i,C_\bx) \,, \qquad i=1,2 \,,
    \end{align*}
    where the last equality is on behalf of \eqref{eq:by_def}-\eqref{eq:G_to_sum}.
\end{proof}

\begin{remark}
      In fact, full triagularization is not needed. It would have been sufficient to achieve a block-triangular structure, where the interference between the last $N_t-1$ channels is arbitrary (conserving the determinant of the block in $T_i$). However, as indicated at the end of \secref{sec:GGTD}, this does not allow to relax the condition \eqref{JSCC_condition}: a block corresponding to the last $N_t - 1$ channels may always be replaced by a triangular block with a constant diagonal (see \remref{rem:FixedRatio}). Moreover, the triangular form is advantageous from the point of view of complexity (see \secref{sec:GGTD}).
\end{remark}

Theorem~\ref{thm:JSCC} does not imply that $\bar{\mS}(H_1,H_2)$ is fully achievable, since the conditions on the GSVs should be verified separately for each optimal input covariance matrix $C_\bx$. However, in the sequel we show that for $N_t\leq 2$, the condition can be verified directly on the channel matrices $H_1$ and $H_2$. Similarly, if the channel matrices are of (any) proper dimensions, then at the limit of high SNR (as the choice $C_\bx=I$ becomes optimal), the GSVs of the augmented matrices approach those of $(H_1,H_2)$, thus the condition may be applied to the channel matrices directly, verifying achievability of the whole region at once.

\subsection{Two Transmit Antennas}
\label{sub:rank2}

In this section we consider the case where $N_t = 2$. In that case, the GSV vector $\bmu(H_1,H_2)$ has two elements. We say that the GSV vector is \emph{mixed}, if one of the elements is at least one, and the other is at most one.
The following is proven in Appendix~\ref{app:rank2}.

\vspace{3mm}
\begin{lemma}
\label{lemma_rank2}
  Let $H_1$ and $H_2$ be two matrices of proper dimensions, with two columns, and define the augmented matrices (as in \eqref{augmented}):
  \[G_i = \left(  \begin{array} {c} H_i C^{1/2} \\ I_2 \end{array} \right) \,, \]
  where $C$ is some Hermitian positive semi-definite matrix.
  Then if $\bmu(H_1,H_2)$ is mixed, $\bmu(G_1,G_2)$ is mixed as well.
\vspace{2mm}
\end{lemma}

We use this lemma and \thmref{thm:JSCC} to prove the following.

\vspace{3mm}
\begin{corol}
Let $H_1,H_2$ be channel matrices with $N_t=2$. If $\bmu(H_1,H_2)$ is mixed, then the bounding region $\bar{\mS}(H_1,H_2)$ of Proposition~\ref{prop:bound} is achievable.
\vspace{2mm}
\end{corol}

\begin{proof}
For any point on the boundary of  $\bar{\mS}(H_1,H_2)$, let $\bmu$ be the GSV vector of the augmented matrices with the corresponding $C_\bx$. By Lemma~\ref{lemma_rank2}, $\bmu$ is mixed as well. Now if the product of $\bmu$ is at most one, we can apply \thmref{thm:JSCC}. If it is greater than one, we switch the indices between $H_1$ and $H_2$, and then apply \thmref{thm:JSCC}.
\end{proof}

Unfortunately, this result cannot be generalized to the case $N_t>2$: Although at any dimension it remains true that the number of GSVs smaller or greater than one is not changed by the augmentation, this property does not hold for \emph{products} of GSVs as required for applying \thmref{thm:JSCC}.

 In order to demonstrate this result, consider the simplest example of the diagonal two-input two-output case:\footnote{Being diagonal, this channel may be obtained from a single-input single-output Gaussian inter-symbol interference channel which has a two-step frequency response, by applying the discrete Fourier transform.}

\beq{eq:two-step-as-MIMO}
  H_i = \left(  \begin{array} {c c}
		  \alpha_i & 0 \\ 0 & \beta_i
		\end{array}\right)
  \,, \quad i=1,2 \,.
\eeq

     The bounding SDR region  $\bar{\mS}(H_1,H_2)$ now becomes:
    \begin{align}
    \label{eq:two-step-bound}
        \bigcup_{0\leq\gamma\leq 1} \Big\{(\SDR_1,\SDR_2) &:  \\
    \nonumber
        &\hspace{-2cm} \SDR_i \leq \Bigl(1+|\alpha_i|^2\gamma P\Bigr)  \Bigl(1+|\beta_i|^2 (1-\gamma) P\big) \Big\}.
    \end{align}
    In this expression, $\gamma$ is the portion of the transmit power sent over the first band.

    We point out a few special cases where points on the surface of this region are achievable by known strategies.

    \begin{enumerate}
     \item
        \textbf{No BW expansion: Analog transmission.} If one of the bands has zero capacity, e.g., $\beta_1=\beta_2=0$, \eqref{eq:two-step-bound} reduces to: $\SDR_i \leq 1+ |\alpha_i|^2 P$, which is
        achievable via analog transmission \cite{Goblick65}. If for each user a different band is usable, e.g., $\alpha_1=\beta_2=0$, any transmission (digital or analog) which is orthogonal between users is optimal.

    \item
        \textbf{Equal SDRs: Digital transmission.} A point on the boundary which satisfies $\SDR_1=\SDR_2$ may be achieved by quantizing the source and then using a digital common-message (multicasting) code for the BC channel (as described in \secref{subs:Mutlicasting}.

    \item
        \textbf{One equal band: HDA transmission.} If for one of the bands the gains are equal, e.g., $|\beta_1|=|\beta_2|=\beta$, we can use that band for digital transmission with rate $R_\digital = \log(1+\beta^2 P)$ and then apply Proposition~\ref{prop:Mittal} to achieve the bound \eqref{eq:two-step-bound}.
    \end{enumerate}

\begin{figure}[t]
   \includegraphics[width=\columnwidth]{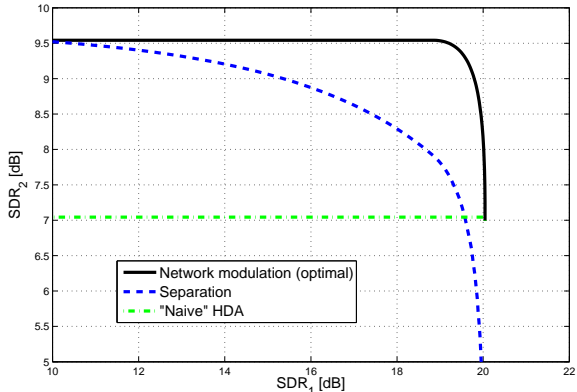}
   \vspace{-10mm}
   \caption{Performance comparison for $\alpha_1=1$, $\beta_1=10$,$\alpha_2=\beta_2=2$,$P=1$. }
   \label{fig:performance}
   \vspace{-5mm}
\end{figure}

Using network modulation, we can extend the HDA transmission (case 3 above), by transforming a diagonal channel where none of the gains is equal between users, to an equivalent triangular channel where for one of the bands the gain is equal. This can be done under the condition \eqref{JSCC_condition}, which specializes to (allowing to swap roles between matrices):
 \beq{two-step-condition} |\alpha_1|^2 \geq |\alpha_2|^2 \text{ and } |\beta_1|^2 \leq |\beta_2|^2 \eeq or vice versa. This is an ``anti-degradedness'' condition: No user can have better SNR on both bands. This condition subsumes all the cases mentioned above. It is not known whether it is a necessary condition, but at least for the case where both channels are white ($\alpha_i=\beta_i$), it was shown in \cite{Reznick} that simultaneous optimality is \emph{not} possible.

\figref{fig:performance} shows a numerical evaluation of performance for some gain values. It can be appreciated that the optimal performance imposes almost no tradeoff between users. Indeed the only tradeoff comes from the need to choose the same $C_\bx$. Thus, in the limit of high SNR, both users attain their optimal single-user performance. For comparison, we show the performance of a separation-based scheme, where a successive-refinement source code is transmitted over a digital broadcast channel code, as well as that of a ``na\"ive'' HDA scheme, where transmission is digital over one band and analog over the other.

\section{Conclusions}
\label{sec:conclusions}

This work considered the problem of transmitting analog (source) and digital information over MIMO communication networks.
For this, we proposed a new decomposition which triangularizes two matrices simultaneously using the same unitary transformation on one side, and different~-- on the other, which allowed, in turn, to accommodate a modulation to the network and the desired application. We then showed how using one version of this decomposition it is possible to construct a practical capacity-achieving scheme for two-user multicasting (which may also be useful in different relaying problems), whereas a different version of this decomposition becomes useful when transmitting the same (analog) source over two different MIMO channels. In the latter case, this technique allowed deriving new achievable regions, which proved optimal for a class of channels.


\appendices


\section{Joint Decomposition for Non-Square Matrices}
\label{app:Proof_GGTD_non-square}

In this Appendix we complete the proof of the direct part of \thmref{thm:GGTD}, by considering the general proper-dimension case.

We start by decomposing $A_i$ using the QR decomposition:
\begin{align*}
    A_i = Q_i R_i \,, \qquad i=1,2 \,,
\end{align*}
where $Q_i$ is unitary and $R_i$ is generalized upper-triangular with non-negative diagonal entries. Moreover, the GSV vectors of $(A_1,A_2)$ and $(R_1,R_2)$ are equal, $\bmu(A_1,A_2) = \bmu(R_1,R_2)$, since $A_i$ and $R_i$ are equal up to a unitary transformation on the left.

Since $A_i$ is full-rank and $m_i \geq n$,
the diagonal elements of $R_i$ are all (strictly) positive and the entries on its lower \mbox{$(m_i - n)$} rows are all zeros.
  Note that the square parts $[R_1]$ and $[R_2]$ are non-singular, with $\bmu([R_1],[R_2])=\bmu(R_1,R_2)= \bmu(A_1,A_2)$. Thus $\bmu([R_1],[R_2]) \succeq \br(T_1,T_2)$.
Invoking the proof for the square case in \secref{sec:GGTD}, we may decompose $[R_1]$ and $[R_2]$ simultaneously as:
\begin{align*}
    [R_1] &= \tU_1 \tT_1 V^\dagger \\
    [R_2] &= \tU_2 \tT_2 V^\dagger \,,
\end{align*}
where $\br(\tT_1,\tT_2)=\br(T_1,T_2)$.
Now, construct the augmented unitary matrices $P_i$:
\begin{align*}
    P_i \triangleq \left(
	    \begin{array}{cc}
	      \tU_i & 0 \\
	      0 & I_{m_i-n} \\
	    \end{array}
	  \right) \,,
\end{align*}
and the generalized triangular matrices $T_i$ of dimensions \mbox{$m_i \times n$}:
\begin{align*}
T_i \triangleq
    \left(
      \begin{array}{c}
	\tT_i \\
	0 \\
      \end{array}
    \right) \,.
\end{align*}

Thus, we arrive at the desired decomposition of $A_1$ and $A_2$ \eqref{eq:JointTriang.},
with $U_i \triangleq Q_i P_i$ and $V$.
\hfill $\blacksquare$

\section{Proof of Lemma~\ref{lemma_rank2}}
\label{app:rank2}

We first Note that if $C$ is singular, then at least one of the elements of $\bmu(G_1,G_2)$ equals one, suggesting the latter is mixed.
Thus, we are left with the case of a non-singular matrix $C$. Let $F_i \triangleq H_i C^{1/2}$ for $i=1,2$.
We claim that $\bmu(F_1,F_2)$ must be mixed. This is true, since for a non-singular matrix $C$: $\bmu(F_1,F_2)=\bmu(H_1,H_2)$.
It is left to show that if $\bmu(F_1,F_2)$ is mixed, then so is $\bmu(G_1,G_2)$. To that end, define the quadratic functions:
\begin{align*}
    p(x) &\triangleq \det \left( F_1^\dagger F_1 - x F_2^\dagger F_2 \right) \,, \\
    q(x) &\triangleq \det \left( G_1^\dagger G_1 - x G_2^\dagger G_2 \right). \end{align*}
By Definition~\ref{def:GSV}, the roots of $p(x)$ and $q(x)$ equal the square of the elements of $\bmu(F_1,F_2)$ and $\bmu(G_1,G_2)$, respectively. Thus it suffices to prove that if the roots of $p(x)$ are not on the same side of $x=1$, then so are the roots of $q(x)$.
By the positive semi-definitiveness of $F_i$ and $G_i$, both functions are convex-$\bigcup$ with $p(0),q(0),p(\infty),q(\infty)\geq 0$. By the assumption on the roots of $p(x)$, it must be that $p(1)\leq 0$. But since
\[G_1^\dagger G_1 - G_2^\dagger G_2 = F_1^\dagger F_1 - F_2^\dagger F_2 \,,\] we have that $q(1)=p(1)$, and thus $q(1)\leq 0$. Finally, a convex-$\bigcup$ continuous function which is non-negative at $x=0$ and for $x\rightarrow \infty$, and is non-positive at $x=1$ cannot have both roots on the same side of $1$.
\hfill $\blacksquare$


\bibliography{toly3}
\bibliographystyle{unsrt}

\end{document}